%% file: main.tex
\documentclass[11pt]{article}

\usepackage[english]{babel}

\usepackage[letterpaper,top=2cm,bottom=2cm,left=3cm,right=3cm,marginparwidth=1.75cm]{geometry}

\usepackage{amsmath,amssymb,amsfonts,mathrsfs,amsthm,mathtools, bbm, enumerate}
\usepackage{comment}
\usepackage{fullpage}
\usepackage{graphicx}
\usepackage{xcolor}
\usepackage{lipsum}
\usepackage[font=itshape]{quoting}
\usepackage[bottom]{footmisc}
\usepackage{csquotes}

\allowdisplaybreaks[4]

\usepackage[colorlinks=true,linkcolor=blue,citecolor=blue]{hyperref}

\usepackage[backend=biber, isbn=false, style=alphabetic, backref=true,
doi=false, url=false, maxcitenames=10, mincitenames=5,
maxalphanames=10, maxbibnames=10, minbibnames=5, minalphanames=3,
defernumbers=true, sortlocale=en_US, sorting=ynt, sortcites]{biblatex}
\AtBeginBibliography{\small}
\input{defs}

\newif\ifblind

\blindtrue

\title{Optimal Electrical Oblivious Routing on Expanders}
\author{}
\ifblind
\author{Cella Florescu  \\ ETH Zurich \\ cella.florescu@inf.ethz.ch \and Rasmus Kyng\thanks{The research leading to these results has received funding from the grant ``Algorithms and complexity for high-accuracy flows and convex optimization'' (no. 200021 204787) of the Swiss National Science Foundation.} \\ ETH Zurich  \\ kyng@inf.ethz.ch \and Maximilian Probst Gutenberg\footnotemark[1]\\ ETH Zurich \\ maxprobst@ethz.ch \and Sushant Sachdeva\thanks{Sushant Sachdeva was supported by an NSERC Discovery Grant RGPIN-2018-06398, an Ontario Early Researcher Award (ERA) ER21-16-284, and a Sloan Research Fellowship.}
\\ University of Toronto \\ sachdeva@cs.toronto.edu}
\fi
\date{}

\begin{document}

\setlength{\abovedisplayskip}{5pt}
\setlength{\belowdisplayskip}{5pt}

\pagenumbering{gobble}

\maketitle

\begin{abstract}
In this paper, we investigate the question of whether the electrical flow routing is a good oblivious routing scheme on an $m$-edge graph $G = (V, E)$ that is a $\Phi$-expander, i.e. where $|\partial S| \geq \Phi \cdot \vol(S)$ for every $S \subseteq V, \vol(S) \leq \vol(V)/2$. Beyond its simplicity and structural importance, this question is well-motivated by the current state-of-the-art of fast algorithms for $\ell_{\infty}$ oblivious routings that reduce to the expander-case which is in turn solved by electrical flow routing. 

Our main result proves that the electrical routing is an $O(\Phi^{-1} \log m)$-competitive oblivious routing in the $\ell_1$- and $\ell_\infty$-norms. We further observe that the oblivious routing is $O(\log^2 m)$-competitive in the $\ell_2$-norm and, in fact, $O(\log m)$-competitive if $\ell_2$-localization is $O(\log m)$ which is widely believed. 

Using these three upper bounds, we can smoothly interpolate to obtain upper bounds for every $p \in [2, \infty]$ and $q$ given by $1/p + 1/q = 1$. Assuming $\ell_2$-localization in $O(\log m)$, we obtain that in $\ell_p$ and $\ell_q$, the electrical oblivious routing is $O(\Phi^{-(1-2/p)}\log m)$ competitive. Using the currently known result for $\ell_2$-localization, this ratio deteriorates by at most a sublogarithmic factor for every $p, q \neq 2$.

We complement our upper bounds with lower bounds that show that the electrical routing for any such $p$ and $q$ is $\Omega(\Phi^{-(1-2/p)}\log m)$-competitive. This renders our results in $\ell_1$ and $\ell_{\infty}$ unconditionally tight up to constants, and the result in any $\ell_p$- and $\ell_q$-norm to be tight in case of $\ell_2$-localization in $O(\log m)$. 
\end{abstract}

\pagebreak

\pagenumbering{arabic}

\input{introduction}

\input{prelims}
\input{competitiveratio}
\input{rieszthorin}
\input{lowerbound}

\vspace*{\fill}

\ifblind
\paragraph{Acknowledgement.} We would like to thank Yang P. Liu for pointing us to the Riesz-Thorin theorem.
\fi

\pagebreak

\begin{refcontext}[sorting=nyt]
\printbibliography[heading=bibintoc]
\end{refcontext}

\pagebreak

\appendix

\input{generalizednorms}

\end{document}

%% file: defs.tex
\newcommand{\rasmus}[1]{}
\newcommand{\sushant}[1]{}
\newcommand{\mprobst}[1]{}
\newcommand{\cella}[1]{}

\newcommand{\todo}[1]{\textbf{\color{red}[TODO: #1]}}

\newcommand{\R}{\mathbb{R}}

\newcommand{\matr}[1]{\boldsymbol{#1}}
\newcommand{\vol}{\mathrm{vol}}
\newcommand{\conge}{\mathrm{cong}}
\newcommand{\opt}{\mathrm{opt}}

\DeclareMathOperator{\poly}{poly}

\newcommand{\Otil}{\tilde{O}}

\newcommand{\trp}{\top}

\DeclarePairedDelimiter\abs{\lvert}{\rvert}
\DeclarePairedDelimiter\norm{\lVert}{\rVert}

\renewcommand{\epsilon}{\ensuremath\varepsilon}

\renewcommand{\phi}{\ensuremath{\varphi}}

\newtheorem{theorem}{Theorem}[section]
\newtheorem{corollary}[theorem]{Corollary}
\newtheorem{lemma}[theorem]{Lemma}

\newtheorem{claim}[theorem]{Claim}

\newtheorem{fact}[theorem]{Fact}

\newtheorem{Informal Theorem}[theorem]{Informal Theorem}

\theoremstyle{definition}
\newtheorem{definition}[theorem]{Definition}
\newtheorem{remark}[theorem]{Remark}

\newtheorem*{theorem*}{Theorem}
\newtheorem*{corollary*}{Corollary}
\newtheorem*{conjecture*}{Conjecture}
\newtheorem*{lemma*}{Lemma}
\newtheorem*{thm*}{Theorem}
\newtheorem*{prop*}{Proposition}
\newtheorem*{obs*}{Observation}
\newtheorem*{definition*}{Definition}

\newtheorem*{remark*}{Remark}
\newtheorem*{rec*}{Recommendation}

%% file: introduction.tex
\newcommand{\package}{\emph}

\section{Introduction}

In this paper, we study flow-routing problems on connected, undirected (multi-)graphs $G = (V,E)$.
A broad and well-studied class of single-commodity flow problems arises by seeking a flow $\matr{f} \in \R^E$ that routes given demands $\matr{\chi} \in \R^V$, while minimizing a \(\ell_p\)-norm of the flow. 
Denoting the graph edge-vertex incidence matrix by $\matr{B} \in \R^{V \times E}$,
we can write these optimization problems as 
\begin{equation}
\min_{\matr{B}\matr{f}
= \matr{\chi}}\norm{\matr{f}}_{p}.
\label{equation_singlecomm}
\end{equation}
The case $p = \infty$ is known as undirected maximum flow, 
while $p = 2$ is called electrical flow and $p = 1$ is called transshipment. Here we focus for simplicity on the unweighted setting, but all results in this paper and in related work can in fact be extended to work in weighted graphs.

We can generalize these flow problems to the multi-commodity case by allowing a collection of demands \(\{\matr{\chi}_i\}\)  to be routed simultaneously by a collection of flows
\(\{\matr{f}_i\}\), while minimizing a single objective on all of them.
\begin{equation}
\min_{\matr{B}\matr{f}_i = \matr{\chi}_i, \forall i}
\norm{
\sum_i \abs{\matr{f}_i}}_{p}.
\label{equation_multicomm}
\end{equation}
For any $p$, solutions with $(1+1/\poly(|E|))$-multiplicative error  to these problems can be computed in polynomial time and for the single-commodity setting even in almost-linear time \cite{KPSW19, CKLPGS22}. For the special cases of $p = 1,2,\infty$, optimal solutions can be computed in polynomial time via linear/convex programming.

However, in many settings, we may want to sacrifice optimality of our routing solutions for simplicity of the routing algorithm. 
A particularly simple and popular approach is \emph{oblivious routing}, where a collection of routing paths are chosen in advance between every pair of nodes, without knowing the demands that will be eventually routed.
Historically, oblivious routings were first studied on specific networks, specifically the hypercube~\cite{ValiantB81, valiant1982scheme}. 
 A deeply influential technique in this area is the work of Rack\"{e}~\cite{R02}. 
An \emph{oblivious routing} is linear operator $\matr{A} \in \R^{E \times V}$ that maps any valid\footnote{
A demand $\matr{\chi}$ can be routed on a connected graph iff $\sum_v \chi(v) = 0$.}  demand vector $\matr{\chi} \in \R^V$
to a flow $\matr{f} = \matr{A} \matr{\chi}$ that routes $\matr{\chi}$. This extends to routing multiple demands in the multi-commodity setting,
\(\{ \matr{f}_i = \matr{A} \matr{\chi}_i\}\).

Conceptually, a highly attractive feature is that multiple demand vectors can be routed simultaneously without knowing the other demands, and a single demand can be broken down into multiple terms, e.g. source-sink demand pairs, and routings of each pair can be again computed separately.
These features make oblivious routings ideal for online routing problems -- which was the original motivation for their construction \cite{R02}.


As mentioned above, using oblivious routing comes at the sacrifice of optimality. To get a quantitative measure of the loss a routing scheme $\matr{A}$ might incur in the $\ell_p$ metric, we define the \emph{competitive ratio} of $\matr{A}$, denoted by $\rho_p(\matr{A})$, to be the maximal ratio between the objective value achieved by an oblivious routing $\matr{A}$ and the optimal solution achieved by \emph{any} (multi-commodity) demand. 

In a ground-breaking sequence of papers \cite{R02,ACFKR03,ER09}, R\"{a}cke, Azar, Cohen, Fiat, Kaplan, and Englert showed that for all \(\ell_p\)-norms, oblivious routings with competitive ratio $\Otil(1)$ exist\footnote{
We use $\Otil(\cdot)$ to hide polylogarithmic factors in the graph size $m$.
}. In fact, for the well-studied setting of $p = \infty$, \cite{R08} gave an optimal construction with $O(\log m)$ competitive ratio in polynomial time, matching a $\Omega(\log m)$ lower bound \cite{BL97, MMadHVW97}.

\paragraph*{Fast Algorithms and Applications for $\ell_\infty$ Oblivious Routing} $\ell_\infty$ oblivious routings are a fundamental tool in obtaining fast approximate maximum flow algorithms in undirected graphs. Building on the techniques in \cite{ST04, M10}, \cite{KLOS, S13} give algorithms that show that $O(\poly(\alpha/ \epsilon))$ applications of a $\alpha$-competitive $\ell_\infty$ oblivious routing yield $(1+\epsilon)$-approximate maximum flow on undirected graphs by using gradient descent. 
They then developed almost-linear time algorithms for $m^{o(1)}$-competitive $\ell_\infty$ oblivious routing. 
As a result, they obtained approximate undirected maximum flow in time $m^{1+o(1)} \poly(1/\epsilon)$ -- one of the major recent breakthroughs in modern graph algorithms. 

Later, \cite{RST14} gave a reduction from computing $\Otil(1)$-competitive $\ell_\infty$ oblivious routings to approximate maximum flows resulting in a $m^{1+o(1)}$ time algorithm.  \cite{P16} then showed that combining these approaches recursively yields a $\Otil(m)$ algorithm to compute $\Otil(1)$-competitive $\ell_\infty$ oblivious routings and a $\Otil(m\poly(1/\epsilon))$ algorithm for $(1+\epsilon)$-approximate undirected maximum flow \cite{P16}. 
Recently, \cite{goranci2021expander} presented an alternative, simple $\Otil(m)$ algorithm to obtain (and maintain) $\ell_{\infty}$ oblivious routings with subpolynomial competitive factor.

While recently the first \emph{exact} maximum flow algorithm with runtime $m^{1+o(1)}$ was given in \cite{CKLPGS22}, $\ell_{\infty}$ oblivious routings and approximate undirected maximum flow remain important tools with many algorithms crucially relying on them as subroutines to obtain runtime $\Otil(m)$.

We point out that above, for simplicity, we did not properly distinguish between $\ell_{\infty}$ oblivious routings which are only constructed in \cite{KLOS}, and their weaker counterparts \emph{congestion approximators} which are used in all other constructions. A congestion approximator is a linear operator $\matr{C}$ that maps each demand $\matr{\chi}$ to vector $\matr{c} = \matr{C}\matr{\chi}$ such that $\|\matr{c}\|_{\infty}$ approximates the objective value of \eqref{equation_multicomm}. Note, that $\matr{c}$ is not necessarily a flow.  


\paragraph*{$\ell_\infty$ Oblivious Routing on Expanders and in General Graphs} 
\emph{Valiant's trick}~\cite{valiant1982scheme}, a popular scheme that routes demands from each source to a set of randomly chosen intermediate nodes before routing them to the destination, establishes the existence of $O(\Phi^{-1}\log n)$-competitive $\ell_{\infty}$-oblivious routings in expanders. However, implementing Valiant's trick algorithmically requires computing multi-commodity flows, which are expensive to compute.

To the best of our knowledge, the only fast algorithm that computes an $\ell_{\infty}$ oblivious routing on general graphs is given in \cite{KLOS}. 
In their approach, they first reduce the problem to finding an $\ell_{\infty}$ oblivious routing on a $\Phi$-expander with unit-weights (in this case $\Phi = m^{-o(1)}$). They then exploit a simple but striking statement, previously demonstrated by Kelner and Maymounkov \cite{KM11}: the electrical flow routing, henceforth denoted by $\matr{A}_{\mathcal{E}}$, on a $\Phi$-expander is a $O(\Phi^{-2} \log m)$-competitive $\ell_{\infty}$-routing. 
It was later observed by Schild-Rao-Srivastava~\cite{SRS18} that on unweighted graphs, the statement can be derived from Cheeger's Inequality (\cite{C70,AM85}).
Further, the electrical flow routing can be applied efficiently after $\Otil(m)$ preprocessing, due to the breakthrough result by Spielman and Teng \cite{ST04} and subsequent work \cite{kelner2013simple, cohen2014solving, kyng2016approximate, jambulapati2021ultrasparse}\footnote{Technically, only a high-accuracy solution is obtained which suffices for our application.}. \cite{KLOS} then demonstrates that by assembling and combining these routings on expanders, one obtains an $\ell_{\infty}$ oblivious routing of the entire graph that can be evaluated efficiently.

As far as we know, no other fast algorithm is currently known to compute $\ell_{\infty}$ oblivious routing, and all fast algorithms that compute congestion approximators again reduce to expanders on which cuts can be approximated by stars. Therefore, to the best of our knowledge, every almost-linear time approach to constructing $\ell_{\infty}$ oblivious routing  reduces to expanders, and on expanders the only known fast algorithm for obtaining an $\ell_{\infty}$ oblivious routing is to use the electrical flow routing.

\paragraph*{Oblivious Routing for any $\ell_{p}$} Analogous to the reduction of solving approximate undirected maximum flow via few applications of $\ell_{\infty}$ oblivious routing, Sherman later showed in \cite{sherman2017generalized}, that any \(\ell_p\)-norm minimizing flow on undirected graphs can be computed to an $(1+\epsilon)$-approximation by applying $\ell_p$ oblivious routings with $\alpha$ competitive ratio $\Otil(\poly(\alpha/\epsilon))$ times via gradient descent. 

While we are not aware of any article studying fast algorithms for the general \(\ell_p\)-norm, the $\ell_1$-norm has received considerable attention and $\Otil(m)$ time algorithms were given with competitive ratio $\Otil(1)$ \cite{li2020faster, ZuzicGYHS22,  rozhovn2022undirected},
and adapted to fully-dynamic graphs in \cite{ChenKLMP23}.

\paragraph*{Oblivious Routing for $\ell_{1}$ on Expanders} Further, at least in the unit-capacity setting, the result by Kelner and Maymounkov \cite{KM11} extends seamlessly to the $\ell_1$-norm, i.e. the electric flow routing $\matr{A}_{\mathcal{E}}$ has competitive ratio $O(\Phi^{-2}\log m)$. This follows since the electrical flow routing is given by $\matr{A}_{\mathcal{E}} = \matr{B}^{\trp} \matr{L}^+$, where $\matr{B}$ is the vertex-edge
incidence matrix and $\matr{L} = \matr{B}\matr{B}^\trp$ is the
Laplacian matrix of the graph and then bounding the $\ell_{\infty}$ competitive ratio of the oblivious routing for multicommodity flow problems is equivalent to bounding the quantity $\norm{\abs{\matr{A}_{\mathcal{E}} \matr{B}}}_{\infty \to \infty}$, where $\abs{\cdot}$ denotes the entrywise absolute value, while the $\ell_{1}$ competitive ratio equals $\norm{\abs{\matr{A}_{\mathcal{E}}\matr{B}}}_{1 \to 1}$. The matrix $\Pi = \matr{A}_{\mathcal{E}} \matr{B} = \matr{B}^{\trp}\matr{L}^+\matr{B}$ is a frequently-studied orthogonal projection matrix and it is a symmetric matrix, since $\matr{B}^{\trp} \matr{L}^+ \matr{B} = \matr{B}^\trp\matr{A}_{\mathcal{E}}^\trp  = (\matr{A}_{\mathcal{E}} \matr{B})^\trp = (\matr{B}^{\trp} \matr{L}^+ \matr{B})^{\trp}$ where we use that $\matr{L}^+$ is symmetric. But it is further well-known that  $\norm{\matr{X}}_{\infty \to \infty} = \norm{\matr{X}^{\trp}}_{1
  \to 1} $ and thus we have that $\norm{\abs{\Pi}}_{\infty \to \infty} = \norm{\abs{\Pi}}_{1 \to 1}$, i.e. the competitive ratios achieved by $\matr{A}_{\mathcal{E}}$ in $\ell_1$- and $\ell_\infty$-norm are equal.

\paragraph*{Beyond Oblivious Routing} The quantity $\norm{\abs{\Pi}}_{\infty \to 1}$  is important in several other contexts: It
captures the so-called \emph{localization} of electrical flow on the
graph \cite{SRS18}.
Localization measures the $\ell_1$-length of the electrical flow
corresponding to a demand placed at two endpoints of an edge, averaged over all edges.
The bound $\rho_1(\matr{A}_{\mathcal{E}}) = O(\Phi^{-2} \log m)$ implies a stronger statement in expanders: For every such edge-demand, the $\ell_1$-length is bounded by $O(\Phi^{-2} \log m)$.
It is known that in general graphs, localization is bounded by
$O(\log^2 m)$ -- but it is open whether $O(\log m)$ holds (a lower
bound of $\Omega(\log m)$ is known) although it is widely believed. From \cite{KM11}, we see that any graph with expansion
$1/o(\sqrt{\log m})$ achieves localization $o(\log^2 m)$.

Localization has been used in a number of contexts, including sampling random spanning trees in almost-linear time \cite{S18}, computing spectral subspace sparsification \cite{LS18} of Laplacian matrices, and building oblivious routings using $\Otil(\sqrt{m})$ electrical flows \cite{GHRSS23}.

An interesting message of our paper is that electrical flow on expanders simultaneously is an excellent $\ell_1$ and $\ell_\infty$ oblivious routing, i.e. it uses flow paths are are both short and low congestion.
A broad theory of expanders that simultaneously allow for short and low-congestion paths has recently been developed in \cite{GHZ21,HRG22}, allowing for other possible trade-offs between length and congestion than those obtained by conventional expanders.

\subsection{Main Contributions}


In this article, we study a simple but important question:
\begin{quoting}
Given \textbf{any} $p \in [1, \infty]$, what is the competitive ratio $\rho_p(\matr{A}_{\mathcal{E}})$ of the electrical flow routing $\matr{A}_{\mathcal{E}}$ on a $\Phi$-expander? 
\end{quoting}

We first settle this question for the important cases when $p \in \{1, \infty\}$ by proving the following theorem that proves an upper bound that is tight up to constant factors.
\begin{theorem}
  \label{thm:mainLinf}
  For a \(\Phi\)-expander multigraph \(G = (V, E) \)
  with edge-vertex incidence matrix $\matr{B}$ and Laplacian $\matr{L}$,
  the electrical routing $\matr{A}_{\mathcal{E}} =
  \matr{B}^{\trp}\matr{L}^+$ has competitive ratios $\rho_\infty$ and $\rho_1$ for multi-commodity
  $\ell_{\infty}$ and $\ell_{1}$ routing both bounded by 
  \[
    \rho_{\infty}(\matr{A}_{\mathcal{E}}),
    \rho_{1}(\matr{A}_{\mathcal{E}})
    \leq 3 \cdot \frac{\log(2m)}{\Phi}
  \]
\end{theorem}

The Riesz-Thorin theorem then gives us a way to smoothly interpolate between the upper bounds of any two $\ell_{p_1}$- and $\ell_{p_2}$-norm competitive ratios to obtain an upper bound on the \(\ell_p\)-norm competitive ratio $\rho_p(\matr{A}_{\mathcal{E}})$ for any $p_1 < p < p_2$. Using a smooth interpolation between our results for $\ell_1$- and $\ell_\infty$-norm, we thus obtain the following more general result.

\begin{theorem}
\label{thm:mainLp}
For a \(\Phi\)-expander multigraph \(G = (V, E) \), and any $p \in [1,\infty]$, we have that the competitive ratio of $\matr{A}_\mathcal{E}$ is
\[
\rho_{p}(\matr{A}_{\mathcal{E}})
\leq 3 \cdot \frac{\log(2m)}{\Phi}.
\]
\end{theorem}

It was proven in \cite{SRS18}, that $\Pi = \matr{B}^{\trp}\matr{L}^+\matr{B}$ satisfies $\norm{\abs{\Pi}}_{2 \to 2} \leq O(\log^2 n)$ (where $\abs{\cdot}$ indicates entry-wise absolute value).
We refer to this quantity $\norm{\abs{\Pi}}_{2 \to 2}$ as \emph{$\ell_2$-localization}.
It is widely believed that $\norm{\abs{\Pi}}_{2 \to 2} \leq O(\log m)$.
Implicit in earlier works, albeit perhaps not widely observed, is that $\rho_2(\matr{A}_{\mathcal{E}}) = \norm{\abs{\Pi}}_{2 \to 2}$ (see Lemma \ref{lemma_competitive_ratio_as_norm}), i.e. the competitive ratio of multi-commodity $\ell_2$ routing is exactly characterized by $\ell_2$-localization.
By interpolating with this norm bound and our bounds on $\rho_1(\matr{A}_{\mathcal{E}})$ and $\rho_{\infty}(\matr{A}_{\mathcal{E}})$, we obtain potentially much stonger bounds on the competitive ratio $\rho_p(\matr{A}_{\mathcal{E}})$.

\begin{corollary}[implied by Riesz-Thorin]
    \label{thm:mainUpperBoundViaLocalization}
For a \(\Phi\)-expander multigraph \(G = (V, E) \), and any $p \in [2,\infty]$ and $q$ given by $1/p + 1/q = 1$, we have that the competitive ratios of $\matr{A}_\mathcal{E}$ for the $\ell_p$ and $\ell_q$ norms are
\[
\rho_{p}(\matr{A}_{\mathcal{E}}), \rho_{q}(\matr{A}_{\mathcal{E}})
\leq \|\abs{\Pi} \|^{2/p}_{2 \to 2} \left( 3 \cdot \frac{ \log(2m)}{\Phi}\right)^{1-2/p}.
\]
\end{corollary}

We complement our upper bounds with strong, unconditional lower bounds. Remarkably, if $\norm{\abs{\Pi}}_{2 \to 2} \leq O(\log m)$, as widely believed, then our bounds are tight up to constants. Even with the currently known fact $\norm{\abs{\Pi}}_{2 \to 2} \leq O(\log^2 n)$, our lower bounds still prove a sublogarithmic gap in every constant $p \neq 2$ and $q$ and optimal $\Phi$ dependency. 

\begin{theorem}
\label{thm:mainLowerBound}
For an infinite number of positive integers $n$ and any $\Phi \in [1/\sqrt[3]{n}, 1]$, for any $p \in [2,\infty]$ and $q$ given by $1/p + 1/q = 1$, we have that
\[
\rho_{p}(\matr{A}_{\mathcal{E}}), \rho_{q}(\matr{A}_{\mathcal{E}})
\geq \Omega\left( \frac{\log m}{\Phi^{1-2/p}}\right).
\]
\end{theorem}

\begin{figure}[h]
\centering
\includegraphics[width=.85\textwidth]{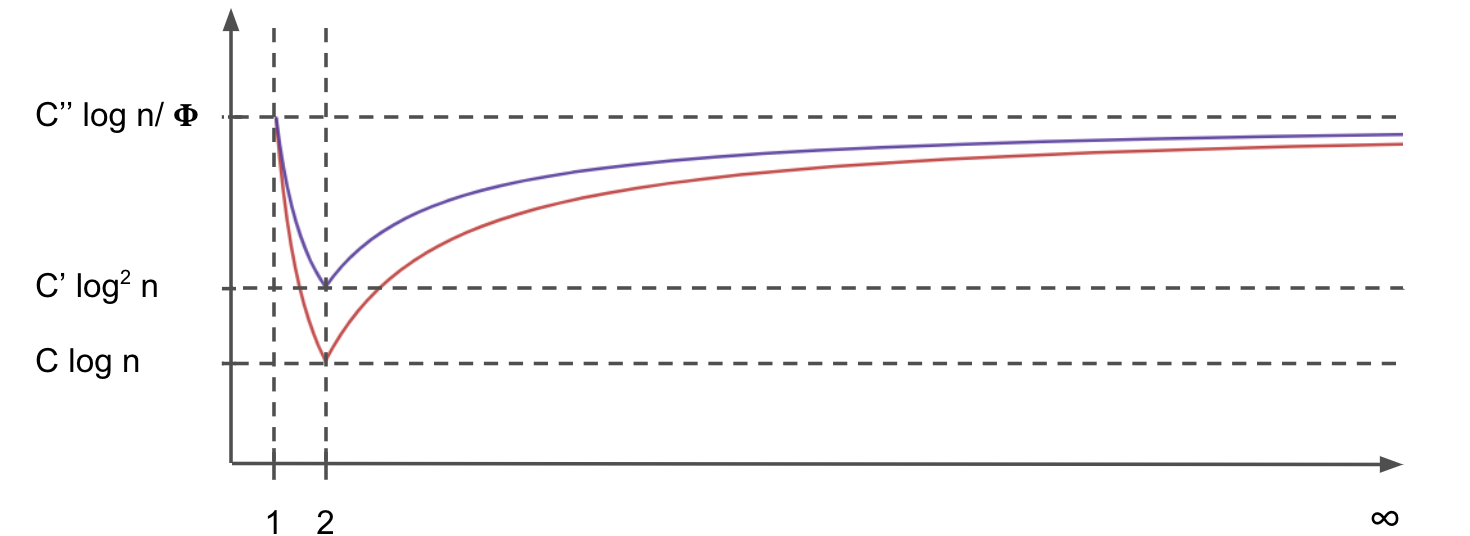}
\caption{ { An illustration of the result given in Corollary \ref{thm:mainUpperBoundViaLocalization} of the different competitive ratios achieved with respect to each $\ell_p$-norm, where $n$ and $\Phi$ are fixed and $\Phi \ll 1/\log m$. The red curve shows the optimal ratios, achieved if localization is in $O(\log m)$, which is also obtained up to constant factors by our lower bound in Theorem \ref{thm:mainLowerBound}. The lilac curve shows the current trade-off where we use the known result that localization is in $O(\log^2 n)$. For $\Phi \gg 1/\log m$ the upper bound of Theorem \ref{thm:mainLp} achieves values between the two curves.}}
\end{figure}

\subsection{Oblivious Electric Routing in Weighted Graphs} 

We can extend Theorem~\ref{thm:mainLinf} to weighted graphs in the
following way: Consider a graph $G = (V,E)$ with positive integer edge
capacities $\matr{c} \in \R^E$ and  positive integer edge lengths
$\matr{s} \in \R^E$.
Letting $\matr{C},\matr{S}  \in \R^{E \times E}$ denote diagonal
matrices with $\matr{c}$ and $\matr{s}$ on the diagonal respectively,
we are interested in the 
optimal weighted $\ell_{\infty}$- and $\ell_{1}$-routings for given
demands
\(D = \{\matr{\chi}_1, \ldots, \matr{\chi}_k\}\)
\begin{equation}
  \opt_{\infty}(D) = \min_{\matr{B}\matr{f}_i = \matr{\chi}_i,
    \forall i} \norm{\matr{C}^{-1} \sum_i \abs{\matr{f}_i}}_{\infty}
  \text{ and }
  \opt_{1}(D) = \min_{\matr{B}\matr{f}_i = \matr{\chi}_i, \forall i}
  \norm{\matr{S} \sum_i \abs{\matr{f}_i}}_{1}
  \label{eq:mcflow}
\end{equation}

Now consider defining an electrical routing by choosing resistances
  $\matr{R} = \matr{S}\matr{C}^{-1}$ and defining the electrical
routing
$\matr{A}_{\mathcal{E}} = \matr{R}^{-1} \matr{B}^{\trp}
(\matr{B}\matr{R}^{-1}\matr{B}^{\trp})^+$.
Let $\hat{G}$ denote the multigraph with edge $e$ replaced by
$\matr{c}(e)$ unit-weight paths of length $\matr{s}(e)$.
Now, one can easily show that the electrical routing in $G$ according
to $\matr{A}_{\mathcal{E}}$ is equal to the unit-weight electrical
routing in $\hat{G}$, when mapping flows on a capacitated edge to a
collection of flows on multi-edge paths.

\begin{corollary}[(Informal) Electrical Oblivious $\ell_{\infty}$- and
  $\ell_1$-Routing on Weighted Expanders]
  \label{cor:informaloblp}
  For (multi-)graph \(G = (V, E) \)
  with edge-vertex incidence matrix $\matr{B}$ and
positive integer edge weights
  and lengths given as diagonal matrices $\matr{C},\matr{S}  \in \R^{E \times E}$,
  the electrical routing $\matr{A}_{\mathcal{E}} = \matr{R}^{-1} \matr{B}^{\trp}
  (\matr{B}\matr{R}^{-1}\matr{B}^{\trp})^+$, where $\matr{R} =
  \matr{S}\matr{C}^{-1}$, 
   has competitive ratios $\rho_\infty$ and $\rho_1$ for multicommodity
   $\ell_{\infty}$ and $\ell_{1}$ routing both bounded by
  \[
    \rho_{\infty}(\matr{A}_{\mathcal{E}}),
    \rho_{1}(\matr{A}_{\mathcal{E}})
    \leq 3 \cdot \frac{\ln(2|E_{\hat{G}}|)}{\Phi(\hat{G})}
  \]
   where $\hat{G}$ denotes the multigraph with edge $e$ replaced by a path of length $\matr{S}(e,e)$ with 
$\matr{C}(e,e)$ unit-weight multi-edges across each hop of the path. 
\end{corollary}

\begin{remark}
  When we take all edge lengths $\matr{S}(e,e)$ to be $1$, the expansion of $\hat{G}$
  equals the usual definition of expansion in graph $G$ with edge
  weight equal to capacity.
\end{remark}

\subsection{New Implications for Localization}

From the connection to localization outlined earlier, we also immediately conclude that localization of graphs with $1/o(\log m)$ expansion improve over the general localization bound of $O(\log^2 m)$.

\begin{corollary}[(Informal) Localization of Electrical Flow]
   For a multigraph \(G = (V, E) \), the average over multi-edges of the $\ell_1$-norm the electrical flow routing 1 unit of flow across the multi-edge is bounded by $O\left(\min\{\Phi^{-1} \log m, \log^2n\}\right)$, and hence graphs with expansion
   $\Phi = 1/o(\log m) $ have localization $o(\log^2 m)$.
\end{corollary}

\subsection{Roadmap}

We next give a Preliminary section to set up the necessary notation for the article. We then prove Theorem \ref{sec:mainInf} in Section \ref{sec:mainInf}. We use this result together with the Riesz-Thorin theorem to obtain Theorem \ref{thm:mainLp} and Corollary \ref{thm:mainUpperBoundViaLocalization} in Section \ref{sec:upperBoundLp}. Finally, in Section \ref{sec:lowerBound}, we give our lower bounds as stated in Theorem \ref{thm:mainLowerBound}.

%% file: prelims.tex
\section{Preliminaries}
\label{section_definitions}

\paragraph*{General Definitions} For any \(n \in \mathbb{N}^*\), we let \([n]\) denote the set \(\{1, 2, \ldots, n\}\). We let \(\matr{1}\) denote the all ones vector and \(\matr{1}_S\) denote the vector that has ones in the positions indexed by the elements of the set \(S\) and zeros otherwise. For any \(\matr{A} \in \R^{m \times n}\), we let \(\lvert \matr{A} \rvert\) denote the matrix where the absolute value operator has been applied entrywise.

\paragraph*{Graphs} Although our results in the contribution section are for unweighted graphs, we also prove stronger statements in the article that also work on weighted graphs. Therefore, we define various notions with respect to weighted graphs. 

Given an input graph \(G = (V, E, w)\) with positive weights which we all assume to be at least $1$, we define $n = |V|$ and $m = |E|$. We assume an arbitrary underlying direction assigned to each edge of \(G\). We define the edge-vertex incidence matrix \(\matr{B} \in \R^{V \times E}\) of \(G\) as 
\[
\matr{B}(w, e) = \begin{cases}
        -1, & \text{if } e = (w, v) \\ 1, & \text{if } e = (v, w) \\ 0, & \text{otherwise}
\end{cases}.
\]
We define the Laplacian $\matr{L} = \matr{B} \matr{W} \matr{B}^\trp$ where $\matr{W}$ is  the diagonal matrix given by the weights \(w\) and denote by \(\matr{L}^+\) the pseudo-inverse of the Laplacian. We call $\Pi = \matr{B}^{\trp} \matr{L}^+ \matr{B}$ the \emph{unweighted} projection matrix of $G$.

\paragraph*{Expanders} We say $G$ is a $\Phi$-expander if $|\partial S| \geq \Phi \cdot \vol(S)$ for every $S \subseteq V, \vol(S) \leq \vol(V)/2$, where we define $|\partial S|$ to be the weight of all edges with exactly one endpoint in $S$, and $\vol(S)$ to be the sum of weighted degrees of vertices in $S$.

\paragraph*{Flows and Congestion} We say \(\matr{\chi} \in \R^V\) is a demand vector if $\matr{1}^\trp \matr{\chi} = 0$. We let \(\matr{\chi}_{(a, b)} \in \R^V\) for every \((a, b) \in E\) to be the unitary demand on the edge \((a, b)\), that is $\matr{\chi}_{(a,b)} = \matr{1}_a - \matr{1}_b$. We say a vector $\matr{f} \in \R^E$ is a flow that routes demand $\matr{\chi}$ if $\matr{B} \matr{f} = \matr{\chi}$. \label{section_definitions_routing} Given an arbitrary norm \(\lVert \cdot \rVert\) on \(\R^E\), we define the congestion of a multi-set of flows \(\{\matr{f}_1, \dots, \matr{f}_k\}\) to be:    \[ \conge\left(\{\matr{f}_1, \dots, \matr{f}_k\}\right) = \left\lVert \matr{W}^{-1} \sum_{i = 1}^k{\lvert \matr{f}_i \rvert}\right\rVert.\]

\paragraph*{Oblivious Routings} We define an oblivious routing on \(G\) to be a linear operator \(\matr{A} \in \R^{E \times V}\) such that \(\matr{B} \matr{A} \matr{\chi} = \matr{\chi}\) for all demand vectors \(\matr{\chi} \in \R^V\), i.e. to be a flow that routes the demand $\matr{\chi}$.

Given a multiset of demands \(D = \{\matr{\chi}_1, \ldots, \matr{\chi}_k\}\), we define  the optimal congestion achievable by 
\[\opt(D) = \min_{\{\matr{f}_i\}_{i \in [k]} \text{ multiset : }\matr{B}\matr{f}_i = \matr{\chi}_i, \forall i}{\conge(\{\matr{f}_i\}_{i \in [k]}}).\]

This allows us to define the competitive ratio of an oblivious routing, which we define
\[\rho(\matr{A}) = \max_{\{\matr{\chi}_i\}_{i \in [k]} \text{ multiset : }\matr{\chi}_i \perp \matr{1}, \forall i}{\frac{\conge\left( \{\matr{A} \matr{\chi}_i\}_{i \in [k]} \right)}{\opt \left(\{\matr{\chi}_i\}_{i \in [k]} \right)}}.\]
Note that whenever we use the subscript ``$p$'' for the competitive ratio $\rho$, we mean that the norm used in defining the congestion in that special case is the $\ell_p$-norm.

\paragraph*{Electrical Flows and Voltages} In this article, we define the electric flow routing operator \(\matr{A}_\mathcal{E} = \matr{W} \matr{B}^\trp \matr{L}^+\).  Right-applying the operator $\matr{A}$ to any demand $\matr{\chi}$ yields the electric flow $\matr{f} = \matr{A} \matr{\chi}$ that routes the demand $\matr{\chi}$. We define the electrical energy associated with the flow vector $\matr{f}$ by $\mathcal{E}(\matr{f}) = \matr{f}^\trp \matr{W}^{-1} \matr{f}$. 

We define the electric voltage vector \(\matr{v} \in \R^V\) with respect to a demand $\matr{\chi}$ by $\matr{v} = \matr{L}^+ \matr{\chi}$. We define the electrical energy associated with the voltage vector \(\matr{v}\) as \(\mathcal{E}(\matr{v}) = \matr{v}^\trp \matr{L} \matr{v}\). Note in particular that the energy of voltages induced by a certain demand coincides with the energy of the respective flow.

We introduce the notion of \textit{``fractional'' volume} at given a threshold \(t \in \R\) with respect to a given voltage vector $\matr{v} \in \mathbb{R}^V$.
We first define the fractional volume per edge and then for the whole graph.
For an edge \((a, b) \in E\): \[\vol_{\geq t}(a, b) =
\begin{cases}
  2 \cdot w(a, b), & \text{if}\ \matr{v}(a) > t \\
  2 \cdot w(a, b) \cdot \frac{\matr{v}(b) - t}{\matr{v}(b) - \matr{v}(a)}, & \text{if}\ \matr{v}(b) \leq t \\
  0, & \text{otherwise}
\end{cases}. \]
For the whole graph \(G\):
\begin{align*}
    \vol_\geq(t) &= \sum_{(a, b) \in E}{\vol_{\geq t}(a, b)}, \\
    \vol_\geq^+(t) &= \vol_\geq(t) + 1.
\end{align*}
Note that in our notation we omit specifying which voltage vector the ``fractional'' volume function is tied to, as it will be clearly specified upon usage during the proofs.

We define \(S_t = \{a \in V \, | \, v(a) \geq t \}\) as the set of vertices whose voltages are greater or equal to the arbitrary threshold \(t \in \R\) and the cut determined by the voltage threshold \(t \in \R\) to be \(C_t = (S_t, V \setminus S_t)\). For convenience, we let the weight of the cut \(C_t\) be \(\delta(t) = \lvert \partial C_t \rvert = \sum_{e \in C_t}{w(e)}\).

%% file: competitiveratio.tex
\section{An Upper Bound on the Competitive Ratio of Electrical Flow Routing for $\ell_{\infty}$}
\label{sec:mainInf}

In this section, we prove our main technical result, Theorem \ref{thm:mainLinf}, by establishing a tight upper bound on the competitive ratio when the congestion is defined in terms of $\ell_{\infty}$ which then by the symmetry of $\Pi$ immediately gives the same competitive ratio in $\ell_1$. 

However, while Theorem \ref{thm:mainLinf} only claims a result for unweighted (multi-)graphs, we show in this section that $\matr{A}_{\mathcal{E}}$ even has good competitive ratio in weighted graphs with respect to $\ell_{\infty}$. However, in the weighted setting, we cannot use the bound for $\ell_{\infty}$ to derive a bound on $\ell_1$, as the matrix-norms that exactly characterize the competitive ratio are not equal in general. Nonetheless, the ``multi-graph'' trick from Corollary \ref{cor:informaloblp} can be used to transform the weighted setting into an unweighted instance and derive bounds.

\paragraph*{Intuition for our Proof}
Kelner and Maymounkov showed that in order to bound the congestion of the electrical routing, it suffices, via a duality (or transposition) argument, to bound the worst case $\ell_1$-norm of the flow induced by routing 1 unit of flow electrically across any edge. We adopt the same approach, but give a more precise analysis.

Suppose $e = (x,y)$ is the edge such that routing one unit of flow between the endpoints causes the highest overall congestion. We let $\matr{v}$ be the associated voltage vector that induces the electrical flow routing one unit from $x$ to $y$. The overall congestion then equals $\sum_{(a,b) \in E} w(a,b) | \matr{v}(a) - \matr{v}(b) |$. We can express this by integrating with respect to voltage along a voltage threshold cut with respect to $\matr{v}$, where the function being integrated at point $t$ is exactly $Y_t = \sum_{(a,b) \in C_t} w(a,b)$, where $C_t$ is the cut at voltage threshold $t$.
This ensures that after integrating $Y_t$ over the entire voltage range, each edge $(a,b)$ contributes exactly $w(a,b) | \matr{v}(a) - \matr{v}(b) |$, as desired. Our  proof proceeds by leveraging that the flow crossing the cut $C_t$ at threshold $t$ is exactly $\sum_{(a,b) \in C_t} w(a,b) | \matr{v}(a) - \matr{v}(b) |$.

As we are sending one unit of flow from $x$ to $y$, and all electrical flow goes one way across a voltage cut, this quantity is exactly 1. At each threshold $t$, this creates an ``on average'' relationship between voltage difference $| \matr{v}(a) - \matr{v}(b) |$ and weight $w(a,b)$ for edges being cut. This in turn allows us to establish a pointwise relationship at each threshold voltage $t$ between the growth in congestion and the change in volume at $t$. Armed with this relationship, we can bound the accumulated congestion of the integrated cuts in terms of the accumulated volume, and this yields our result.

\paragraph*{Contrast with the Kelner-Maymounkov Proof}
It is instructive to consider why the Kelner-Maymounkov congestion bound loses an additional factor $\Phi$ compared to our bound. For concreteness, consider the graph given by a direct edge from $x$ to $y$ and an additional $k$ disjoint paths of length $k$ from $x$ to $y$. It can be shown that in this example, the edge that governs the congestion bound in the strategy above is in fact the direct $(x,y)$ edge.

Kelner-Maymounkov upper bound the true competitive ratio $\rho' = \sum_{(a,b) \in E} w(a,b) | \matr{v}(a) - \matr{v}(b) |$ by the quantity $\rho'' = \sum_{a \in V} \matr{d}(a) |\matr{v}(a) - c|$ for some constant $c$ (see Equation (4.3) in \cite{KM11}). On this concrete graph, $\rho'$ can be explicitly evaluated and is $\Theta(k)$. As the graph has expansion $1 / k$, we can think of this as a bound of $\Theta(1 / \Phi)$. But, $\rho''$ is $\Theta(k^2)$ i.e. $\Theta(1 / \Phi^2)$.

However, Kelner and Maymoukov's strategy makes it difficult to directly bound $\rho'$ as they first measure changes in volume over a (discrete) sequence of threshold cuts, and then changes in voltage over the same sequence of cuts. Their discrete sequence of cuts skips entirely over some edges, i.e. there will be edges that are not crossing any of their cuts. This makes it difficult to establish an estimate for each edge of the pointwise relation between its contribution to volume growth versus voltage growth or congestion growth. Hence, they work with summed bounds on voltage and compare these with summed bounds on volume, which naturally yields bounds on $\rho''$. But, as we have seen, a bound on $\rho''$ must inherently be loose as there is a gap between $\rho'$ and $\rho''$.
\begin{theorem}[$\ell_{\infty}$ Competitive bound of electrical flows]
\label{lemma_improved_competitive_bound}
    For a weighted \(\Phi\)-expander multigraph \(G = (V, E, w) \), the following holds:
    \begin{align*}
    \rho_{\infty}(\matr{A}_\mathcal{E}) \leq 3 \cdot \frac{\ln(\vol(V))}{\Phi}.
    \end{align*}
\end{theorem}

\begin{proof}
We first use that 
\begin{equation}
    \label{equation_rho_as_1norm}
   \rho_{\infty}(\matr{A}_\mathcal{E}) = \max_{e \in E} \, \left\lVert \matr{W} \matr{B}^\trp \matr{L}^+ \matr{\chi}_e \right\rVert_1.
\end{equation}
as shown in \cite{KLOS} as part of the proof of Lemma 26 (by using primarily Lemmas 10 and 11). In order to prove the desired inequality, we then fix an \(e \in E\) such that the quantity in \eqref{equation_rho_as_1norm} gets maximized, 
and let \(\matr{v} = \matr{L}^+ \matr{\chi}_e\) be the voltage induced by setting a unitary demand on this edge. Thus, we equivalently aim to bound:
\begin{equation}
    \label{equation_rho_as_voltage_sum}
 \rho_{\infty}(\matr{A}_\mathcal{E}) = \left\lVert \matr{W} \matr{B}^\trp \matr{L}^+ \boldsymbol{\chi}_e \right\rVert_1 = \left\lVert \matr{W} \matr{B}^\trp \matr{v} \right\rVert_1 = \sum_{(a, b) \in E}{w(a, b) \cdot \lvert \matr{v}(a) - \matr{v}(b)\rvert}.
\end{equation}
Observe now that shifting all the values of \(\matr{v}\) by the same constant does not change the value expressed in \eqref{equation_rho_as_voltage_sum}, and therefore we can assume without loss of generality that the voltages are centered around \(0\), that is we can assume \( \vol\left(\{ i \in V \, | \, \matr{v}(i) \geq 0\}\right) \geq \vol(V) / {2} \) and \( \vol\left(\{ i \in V \, | \, \matr{v}(i) \leq 0\}\right) \geq \vol(V) / {2}\).

Note that, by convention, the electrical flow \(\matr{f}_\mathcal{E} = \matr{A}_\mathcal{E} \matr{\chi}_e\) induces an orientation on the edges in the set \(E\). Henceforth, we assume without loss of generality that orientations of edges in $E$ align with the direction of the electrical flow \(\matr{f}_\mathcal{E}\), that is  \(\matr{f}_\mathcal{E}(a, b) = w(a, b) \cdot (\matr{v}(b) - \matr{v}(a)) \geq 0\) for any \((a, b) \in E\).

In the following, we will employ the definitions introduced in Section \ref{section_definitions}. These concepts give rise to the notion of  ``fractional'' volume, which will ultimately allow us to bound the quantity of Equation \eqref{equation_rho_as_voltage_sum}.

It can easily be proven that \(\vol_\geq\) is continuous in \(\R\) and differentiable at any threshold level \(t \in \R\) for which there does not exist a node \(a \in V\) such that \(\matr{v}(a) = t\). Furthermore, if \(t_\mathrm{min} = \min{\{\matr{v}(a) \, | \, a \in V \}}\) and \(t_\mathrm{max} = \max{\{\matr{v}(a) \, | \, a \in V \}}\), then \(\vol_\geq(t_\mathrm{min}) = \vol(G)\) and \(\vol_\geq(t_\mathrm{max}) = 0\). We can even assume without loss of generality a more precise centering of the voltages around \(0\), namely that \(\vol_\geq(0) = \vol(V) / {2}\).

A voltage threshold level \(t \in \R\) can be seen as naturally determining a cut \(C_t\) in \(G\). Note that the assumption about centering the voltages around \(0\) ensures that \(\vol(S_t) \leq \vol(V) / {2}\) for any \(t > 0\), so it holds that \(\min \{ \vol(S_t), \vol(V \setminus S_t) \} = \vol(S_t)\).

By taking the orientation of the edges into account, we can drop the absolute value operator and rewrite Equation \eqref{equation_rho_as_voltage_sum} as:
\begin{equation}
\begin{aligned}
\label{equation_ratio_as_integral}
\sum_{(a, b) \in E}{w(a, b) \cdot \lvert \matr{v}(a) - \matr{v}(b) \rvert} & = \sum_{(a, b) \in E}{w(a, b) \cdot (\matr{v}(b) - \matr{v}(a))} \\
& = \sum_{(a, b) \in E}{w(a, b) \cdot \int_{t_\mathrm{min}}^{t_\mathrm{max}}{\mathbbm{1}_{\matr{v}(a) < t \leq \matr{v}(b)} \, dt}} \\
& = \int_{t_\mathrm{min}}^{t_\mathrm{max}}{\sum_{(a, b) \in E}{w(a, b) \cdot \mathbbm{1}_{\matr{v}(a) < t \leq \matr{v}(b)}} \, dt} \\
& = \int_{t_\mathrm{min}}^{t_\mathrm{max}}{\delta(t) \, dt} \\
& = \int_{t_\mathrm{min}}^0{\delta(t) \, dt} + \int_0^{t_\mathrm{max}}{\delta(t) \, dt}.
\end{aligned}
\end{equation}
We will bound the quantity in \eqref{equation_ratio_as_integral} by separately bounding each of the two terms in the last equality. Only the proof for the integral over the non-negative values of \(t\) will be presented, the one for the non-positive values proceeds in an analogous manner. Assume thus for the rest of the proof that \(t \geq 0\) holds.

In order to obtain the bound on \(\delta(t)\) for non-negative \(t\), we will inspect the rate of change of the ``fractional'' volume with respect to the voltage threshold of the fixed voltage vector $\matr{v}$. Intuitively, we grow an electrical threshold ball, and directly relate the change in volume to the stretch
accumulated at the current voltage threshold. In more precise terms, we compute a bound on the derivative of \(\vol_\geq^+\) with respect to \(t\) on the domain of differentiability as follows:
\begin{equation}
\begin{aligned}
\label{equation_derivative_of_vol}
-\frac{d}{dt}\vol_\geq^+(t) &= -\frac{d}{dt}\vol_\geq(t) \\
& = -\frac{d}{dt}\left(\sum_{(a, b) \in E} {\vol_{\geq t}(a, b)}\right) \\
& = -\frac{d}{dt}\left(\sum_{(a, b) \in C_t}{\vol_{\geq t}(a, b)}\right) \\
& = \sum_{(a, b) \in C_t}{-\frac{d}{dt}\vol_{\geq t}(a, b)} \\
& = \sum_{(a, b) \in C_t}{2 \frac{w(a, b)}{\matr{v}(b) - \matr{v}(a)}}.
\end{aligned}
\end{equation}
By the construction based on the voltage levels, all edges of the cut \(C_t\) have their head in the set \(S_t\). Therefore, the flow carried by these edges has to be the unit flow, since this is the demand of \(\matr{\chi}_e\):
\begin{equation}
    \label{equation_sum_of_voltages_eq_1}
    1 = \sum_{(a, b) \in C_t}{f_\mathcal{E}(a, b)} = \sum_{(a, b) \in C_t}{w(a, b) \cdot (\matr{v}(b) - \matr{v}(a))}.
\end{equation}

Hereafter, we show that the
negative volume change must exceed the square of the cut size. Informally, the change in volume per edge
is relatively large whenever the voltage gap across the edge is small (volume change scales as
inversely proportional to the gap compared to the cut-value of the edge). But, a ``typical'' edge in the cut must have a fairly small voltage gap, as we otherwise route too much flow across the gap. Formally, since the voltage drops among the edges in the cut \(C_t\) are non-negative, we can use \eqref{equation_sum_of_voltages_eq_1} and the definition of the conductance of \(G\) to further bound \eqref{equation_derivative_of_vol} using the Cauchy–Bunyakovsky–Schwarz inequality:
\begin{equation}
\begin{aligned}
\label{equation_derivative_phi_delta}
    -\frac{d}{dt}\vol_\geq^+(t) &= 2 \sum_{(a, b) \in C_t}{\frac{w(a, b)}{\matr{v}(b) - \matr{v}(a)}} \\
    & = 2 \left(\sum_{(a, b) \in C_t}{\frac{w(a, b)}{\matr{v}(b) - \matr{v}(a)}} \right) \cdot 1\\
    & = 2 \left(\sum_{(a, b) \in C_t}{\frac{w(a, b)}{\matr{v}(b) - \matr{v}(a)}} \right) \left( \sum_{(a, b) \in C_t}{w(a, b) \cdot (\matr{v}(b) - \matr{v}(a))} \right) \\
    &\geq 2 \left(\sum_{(a, b) \in C_t}{\sqrt{\frac{w(a, b)}{\matr{v}(b) - \matr{v}(a)} \cdot w(a, b) \cdot (\matr{v}(b) - \matr{v}(a))}} \right)^2\\
    & \geq 2 \left(\sum_{(a, b) \in C_t}{w(a, b)} \right)^2\\
    & = 2 \cdot \delta(t)^2\\
    & \geq 2 \cdot \delta(t) \cdot \Phi \cdot \vol(S_t).
\end{aligned}
\end{equation}
Denote by \(\vol_\mathrm{int}(t) = \vol(S_t) - \delta(t)\) twice the weight of the edges that have both endpoints in the set \(S_t\). Recall that we assumed all of the edges to have weights at least \(1\), therefore it holds that \(\delta(t) \geq 1\). The definition of ``fractional'' volume implies \(\vol_\geq(t) \leq \vol_\mathrm{int} + 2\delta(t)\), which can be used to further bound \eqref{equation_derivative_phi_delta}:
\begin{align*}
    -\frac{d}{dt}\vol_\geq^+(t) & \geq 2 \cdot \delta(t) \cdot \Phi \cdot \vol(S_t) \\
    & = 2 \cdot \delta(t) \cdot \Phi \cdot (\vol_\mathrm{int}(t) + \delta(t)) \\
    & = \delta(t) \cdot \frac{2\Phi}{3} \cdot (3 \vol_\mathrm{int}(t) + 3\delta(t)) \\
    & \geq \delta(t) \cdot \frac{2\Phi}{3} \cdot (\vol_\mathrm{int}(t) + 2\delta(t) + 1) \\
    & \geq \delta(t) \cdot \frac{2\Phi}{3} \cdot (\vol_\geq(t) + 1) \\
    & = \delta(t) \cdot \frac{2\Phi}{3} \cdot \vol_\geq^+(t).
\end{align*}
Observe that we can rewrite the inequality above to obtain a bound on \(\delta(t)\):
\begin{equation}
    \label{equation_delta_bound}
    \delta(t) \leq -\frac{3}{2 \Phi} \cdot \frac{1}{\vol_\geq^+(t)} \cdot \frac{d}{dt}\vol_\geq^+(t).
\end{equation}

We can now use equation \eqref{equation_delta_bound} to bound the integral over the interval \([0, t_\mathrm{max}]\) in \eqref{equation_ratio_as_integral}:
\[
    \int_0^{t_\mathrm{max}}{\delta(t) \, dt} \leq \int_0^{t_\mathrm{max}}{-\frac{3}{2 \Phi} \cdot \frac{1}{\vol_\geq^+(t)} \cdot \frac{d}{dt}\vol_\geq^+(t) \, dt}.
\]
The latter integral can easily be computed via the change of variable \(u = \vol_\geq^+(t)\), yielding the integration bounds \(\vol_\geq^+(t_\mathrm{max}) = \vol_\geq(t_\mathrm{max}) + 1 = 1\) and \(\vol_\geq^+(0) = \vol_\geq(0) + 1 = \vol(V) / {2} + 1\):
\begin{align*}
    \int_0^{t_\mathrm{max}}{\delta(t) \, dt} & \leq \int_0^{t_\mathrm{max}}{-\frac{3}{2 \Phi} \cdot \frac{1}{\vol_\geq^+(t)} \cdot \frac{d}{dt}\vol_\geq^+(t) \, dt} \\
    & = \int_{t_\mathrm{max}}^0{\frac{3}{2 \Phi} \cdot \frac{1}{\vol_\geq^+(t)} \cdot \frac{d}{dt}\vol_\geq^+(t) \, dt} \\
    & = \frac{3}{2 \Phi} \int_1^{\vol(V) / {2} + 1}{\frac{1}{u} \, du}.
\end{align*}
Since the graph has by assumption at least two nodes connected by an edge with weight at least \(1\), it follows that \(\vol(V) \geq 2 \iff \vol(V) \geq \vol(V) / {2} + 1\). Coupled with the fact that \(u > 0\) for \(u \in [1, \vol(V)]\), this gives us the final bound for the integral:
\begin{align*}
    \int_0^{t_\mathrm{max}}{\delta(t) \, dt} &\leq \frac{3}{2 \Phi} \int_1^{\vol(V) / {2} + 1}{\frac{1}{u} \, du} \\
    &\leq \frac{3}{2 \Phi} \int_1^{\vol(V)}{\frac{1}{u} \, du} \\
    & = \frac{3}{2 \Phi} \cdot \ln(\vol(V)).
\end{align*}
As already mentioned, the same bound can be obtained for the other term in \eqref{equation_ratio_as_integral} in an analogous manner.

Combining the aforementioned result with the relations given by \eqref{equation_rho_as_voltage_sum} and \eqref{equation_ratio_as_integral} gives the desired inequality:
\[
\rho_{\infty}(\matr{A}_\mathcal{E}) = \int_{t_\mathrm{min}}^0{\delta(t) \, dt} + \int_0^{t_\mathrm{max}}{\delta(t) \, dt} \leq 2 \cdot \frac{3 \ln(\vol(V))}{2\Phi} = \frac{3\ln(\vol(V))}{\Phi}.
\]
\end{proof}

%% file: rieszthorin.tex
\section{An Upper Bound on the Competitive Ratio of Electrical Flow Routing for $\ell_{p}$ (for any $p$)}
\label{sec:upperBoundLp}

In this section, we prove two generalizations of Theorem \ref{lemma_improved_competitive_bound}. Previously, we gave a bound on the competitive ratio when the congestion was defined in terms of the $\ell_{\infty}$-norm. This result can be extended to \(\ell_p\)-norms for an arbitrary \(p \in [1, \infty]\) by using Theorem \ref{lemma_improved_competitive_bound}, and instantiating a special case of the Riesz-Thorin theorem. 

This establishes both the results in Theorem \ref{thm:mainLp} and in Corollary \ref{thm:mainUpperBoundViaLocalization}. We stress that the results obtained in this section crucially exploit the symmetry of $\Pi$ and therefore only hold for unweighted graphs.

\paragraph{A Toolbox for $\ell_p$-Norms.} We first use the following result that we prove to much broader generality in Appendix \ref{sec:congestionAbsNorm}.

\begin{lemma}[Competitive ratio of $\ell_p$-norms]
\label{lemma_competitive_ratio_as_norm}
Let \(G = (V, E)\) be a multigraph. For any $p \in [1, \infty]$ and oblivious routing \(\matr{A}\), we have
\[
\rho_p(\matr{A}) = \| \lvert \matr{A} \matr{B}\rvert \|_{p \to p}.
\]
\end{lemma}

Our second tool is the Riesz-Thorin theorem. We explicitly state the two relevant special cases of the theorem that we require in the next section for the convenience of the reader.

\begin{theorem}[Special cases of the Riesz-Thorin theorem, see {{\cite[Theorem 1.3]{SW71}}}]
    \label{lemma:special_case_riesz_thorin}
    Let \(\matr{A} \in \R^{m \times n}\) be a matrix with non-negative entries. For any \(p \in (1, \infty)\) it holds that:
    \begin{align*}
     \lVert \matr{A} \rVert_{p \to p} \leq \lVert \matr{A} \rVert_{1 \to 1}^\frac{1}{p} \cdot \lVert \matr{A} \rVert_{\infty \to \infty}^{1 - \frac{1}{p}}.
    \end{align*}
Furthermore, for $p \in (2, \infty)$,
    \begin{align*}
         \lVert \matr{A} \rVert_{p \to p} \leq \lVert \matr{A} \rVert_{2 \to 2}^\frac{2}{p} \cdot \lVert \matr{A} \rVert_{\infty \to \infty}^{1 - \frac{2}{p}}.
    \end{align*}
\end{theorem}

\paragraph*{Proof of Theorem \ref{thm:mainLp}} We have that the theorem already holds for $\ell_1$ and $\ell_{\infty}$ since we have proven Theorem \ref{thm:mainLinf} in the previous section. Consider therefore any $p \in (1, \infty)$. Then, we have
\begin{align*}
    \rho_{p}(\matr{A}_\mathcal{E}) & \overset{\text{Lemma \ref{lemma_competitive_ratio_as_norm}}}{=} \left\lVert \lvert \matr{A}_\mathcal{E} \matr{B} \rvert \right\rVert_{p \to p} \\
    & \overset{\text{Theorem \ref{lemma:special_case_riesz_thorin}}}{\leq} \lVert \lvert \matr{A}_\mathcal{E} \matr{B} \rvert \rVert_{1 \to 1}^\frac{1}{p} \cdot \lVert \lvert \matr{A}_\mathcal{E} \matr{B} \rvert \rVert_{\infty \to \infty}^{1 - \frac{1}{p}} \\
    & \overset{\text{Lemma \ref{lemma_competitive_ratio_as_norm}}}{=} \left( \rho_{1}(\matr{A}_\mathcal{E}) \right)^\frac{1}{p} \cdot \left( \rho_{\infty}(\matr{A}_\mathcal{E}) \right)^{1-\frac{1}{p}} \\
    &\overset{\text{Theorem \ref{thm:mainLinf}}}{\leq} \left( 3 \cdot \frac{\ln(2m)}{\Phi} \right)^\frac{1}{p} \cdot \left( 3 \cdot \frac{\ln(2m)}{\Phi} \right)^{1-\frac{1}{p}} \\
    &= 3 \cdot \frac{\ln(2m)}{\Phi}.
\end{align*}

\paragraph*{Proof of Corollary \ref{thm:mainUpperBoundViaLocalization}}

Consider next any $p \in (1, \infty)$. Since we have again that for $q$ given by $1/p + 1/q = 1$, we have $\| \matr{X} \|_{p \to p} = \| \matr{X}^{\trp} \|_{q \to q}$, we can assume w.l.o.g. that $p \geq 2$. Similarly to \cite{lawler2009mixing}, we obtain
\begin{align*}
    \rho_{p}(\matr{A}_\mathcal{E}) & \overset{\text{Lemma \ref{lemma_competitive_ratio_as_norm}}}{=} \left\lVert \lvert \matr{A}_\mathcal{E} \matr{B} \rvert \right\rVert_{p \to p} \\
    & \overset{\text{Theorem \ref{lemma:special_case_riesz_thorin}}}{\leq} \lVert \lvert \matr{A}_\mathcal{E} \matr{B} \rvert \rVert_{2 \to 2}^\frac{2}{p} \cdot \lVert \lvert \matr{A}_\mathcal{E} \matr{B} \rvert \rVert_{\infty \to \infty}^{1 - \frac{2}{p}} \\
    & \overset{\text{Lemma \ref{lemma_competitive_ratio_as_norm}}}{=} \left( \rho_{2}(\matr{A}_\mathcal{E}) \right)^\frac{2}{p} \cdot \left( \rho_{\infty}(\matr{A}_\mathcal{E}) \right)^{1-\frac{2}{p}} \\
    &\overset{\text{Lemma \ref{lemma_competitive_ratio_as_norm}, Theorem \ref{thm:mainLinf}}}{\leq} \left( \| \abs{\Pi} \|_{2 \to 2} \right)^\frac{2}{p} \cdot \left( 3 \cdot \frac{\ln(2m)}{\Phi} \right)^{1-\frac{2}{p}}.
\end{align*}


%% file: lowerbound.tex
\section{A Lower Bound for Competitive Ratio of Electric Flow Routing}
\label{sec:lowerBound}

Finally, in this section, we provide a strong lower bound on the competitive ratio of the electrical routing scheme in any $\ell_p$-norm.

\begin{theorem}[Restatement of Theorem \ref{thm:mainLowerBound}]
\label{lemma_improved_lower_bound}
For an infinite number of positive integers $n$ and any $\Phi \in [1/\sqrt[3]{n}, 1]$, for any $p \in [2,\infty]$ and $q$ given by $1/p + 1/q = 1$, we have that
\[
\rho_{p}(\matr{A}_{\mathcal{E}}), \rho_{q}(\matr{A}_{\mathcal{E}})
\geq \Omega\left( \frac{\log m}{\Phi^{1-2/p}}\right).
\]
\end{theorem}

In our proof, we use the following theorem given in \cite{alon2021high}. We remind the reader that the \emph{girth} of a graph $G$ is the weight of the smallest weight cycle of $G$.

\begin{theorem}[{{\cite[Theorem 1.2]{alon2021high}}}]
\label{thm:constructDetExpander}
There are infinitely many positive integer $\Delta$ and $n$, for which an $n$-vertex unweighted graph $G_{\Delta, n} = (V,E)$ exists such that $G$ is $\Phi_{const}$-expander that is $\Delta$-regular with $\Phi_{const} = \Theta(1)$ such that $G$ has girth $\Omega(\log_{\Delta} n)$. 
\end{theorem}

In our proof, we use the existential result behind the statement to refine a proof technique previously used by Englert and Räcke \cite{ER09} to give a lower bound on the competitive ratio of any $\ell_p$ oblivious routing scheme. Our refinement can also be used to strengthen their result by a $\Theta(\log\log n)$ factor.

In our proof, we crucially exploit the following facts about effective resistance. Recall that the effective resistance of a graph $G$ for a pair $(s,t) \in V^2$ is the minimum energy required to route one unit of demand from $s$ to $t$ in $G$, or alternatively the difference in voltages of $s$ and $t$ induced by routing this unit of demand via an electrical flow which is given by $\chi_{(s,t)} \matr{L}^+ \chi_{(s,t)}$. The facts below can be derived straightforwardly from Cheeger's Inequality, mixing of random walks, and characterization of effective resistance by commute times (see for example \cite{kyng2021advanced}).

\begin{fact}\label{fact:constEffResistance}
For $G$ being a constant-degree $\Omega(1)$-expander, we have that the effective resistance of any pair $(s,t) \in V$ is in $\Theta(1)$.
\end{fact} 
\begin{fact}\label{fact:unionOfExpandersForRouting}
Given two constant-degree graphs $G$ and $H$ over the same vertex set $V$. If the effective resistance for a pair $(s,t) \in V^2$ is in $\Theta(1)$ in both $G$ and $H$, then the electrical flow routing one unit of demand from $s$ to $t$ on the union of graphs $G \cup H$ sends at least a constant fraction of the flow over $G$ and a constant fraction of the flow over $H$. 
\end{fact} 

Let us now give a lower bound for any $p \geq 2$ and any parameter $\Phi \in [1/n, 1]$ such that $1/\Phi$ is integer. We start by considering the electrical routing $\matr{A}_{\mathcal{E}}$ for a large constant $\Delta$ and any $n$ and $G_{\Delta, n}$ of the multi-commodity demand that is given by routing for each edge $e = (u,v)$ in $G_{\Delta, n}$ one unit of a commodity from $u$ to $v$, i.e. we consider the demand $\matr{\chi} = \left\{ \matr{\chi}_{(u,v)} \right\}_{e = (u,v) \in E(G_{\Delta, n})}$. Towards understanding the electrical routing, we prove the following simple claim.

\begin{claim}
For any edge $e$ in $G_{\Delta, n}$ where $\Delta$ is a large constant, we have that the electrical flow $\matr{f} = \matr{A}_{\mathcal{E}} \matr{\chi}_{(u,v)}$ routing the demand $ \matr{\chi}_{(u,v)}$ has $\| \matr{f} \|_1 = \Omega(\log n)$. 
\end{claim}
\begin{proof}
The claim follows from showing that $\matr{f}(e)$ carries only $(1-\epsilon)$ units of flow for some constant $\epsilon > 0$. This is because it implies that a constant fraction of the flow is not routed via the edge $e$. But since each path between the endpoints of $e$ that does not use the edge $e$ is of length $\Omega(\log n)$ (by the girth bound in Theorem \ref{thm:constructDetExpander}), we have that this $\epsilon$-fraction adds $\Omega(\epsilon \log n) = \Omega(\log n)$ units of flow to the network $G_{\Delta, n}$.

To prove the claim, it suffices to observe that the graph $G_{\Delta, n} \setminus \{e\}$ is a $\Omega(1)$-expander. But to this end, it suffices to observe that since the conductance of $G_{\Delta, n}$ does not depend on $\Delta$ by Fact \ref{fact:constEffResistance}, by choosing $\Delta$ sufficiently large (i.e. at least twice the inverse of the conductance), we have that each cut contains at least $2$ edges and thus the conductance of $G_{\Delta, n} \setminus \{e\}$ is at least half of the conductance of $G_{\Delta, n}$, and thus still constant.

Using that the trivial graph consisting only of the edge $e$ is a constant-degree $\Omega(1)$-expander, we thus have that the effective resistance of the pair $(u,v)$ in both the graph $G_{\Delta, n} \setminus e$ and $e$ is constant by Fact \ref{fact:constEffResistance}. Thus, by Fact \ref{fact:unionOfExpandersForRouting}, we have that a constant fraction of the demand $\matr{\chi}_{(u,v)}$ is not routed into $e$, as desired.
\end{proof}

Using that multi-commodity flows do not cancel, we thus have that each edge in $G_{\Delta, n}$ carries on average $\Omega(\log n)$ units of flow. We next transform the graph $G_{\Delta, n}$ to then obtain our final gadget on which we can prove the lower bound.

\begin{definition}
Let $G^{\Phi}_{\Delta, n}$ be the graph obtained from $G_{\Delta, n}$ by replacing each edge with $1/\Phi$ vertex-disjoint paths of length $1/\Phi$ between the endpoints of the vertices. Thus, $G^{\Phi}_{\Delta, n}$, for $\Delta$ being a constant, has $\Theta(n/\Phi^2)$ vertices. 
\end{definition}

Next, we claim that the effective resistance of our demand pairs is the same up to a constant in $G_{\Delta, n}$ and $G^{\Phi}_{\Delta, n}$. 

\begin{claim}\label{clm:gadgetAlsoHasGoodRes}
For each edge $(u,v) \in E(G_{\Delta, n})$, the effective resistance of the pair $(u,v)$ in the graph $G^{\Phi}_{\Delta, n}$ is $\Theta(1)$.
\end{claim}
\begin{proof}
To show this result, we give an explicit mapping of the electrical flow routing $\matr{\chi_{(u,v)}}$ in $G_{\Delta, n}$ to routing the flow in $G^{\Phi}_{\Delta, n}$ whose energy is at most constant. Let $\matr{f}$ be this electrical flow routing on $G_{\Delta, n}$, then we map the flow on each edge $e'$ in $G_{\Delta, n}$ uniformly through the $1/\Phi$ disjoint paths between the endpoints of $e'$ in $G^{\Phi}_{\Delta, n}$. Since each path now routes only a $\Phi$-fraction of the original flow on the edge $e'$, we have that the energy used to route through each edge on the disjoint paths replacing $e'_{middle}$ is $(\matr{f}(e') \Phi)^2 = \matr{f}(e')^2  \Phi^2$. We thus have that the energy incurred by routing through the $1/\Phi$ disjoint paths each consisting of $1/\Phi$ edges is $1/\Phi^2 \cdot \matr{f}(e')^2 \Phi^2 = \matr{f}(e')^2$. Thus, the effective resistance of $(u,v)$ in  $G^{\Phi}_{\Delta, n}$ is at most the resistance in $G_{\Delta, n}$ which implies it is in $O(1)$. 

A lower bound of $\Omega(1)$ is observed by inversing this mapping to collect the amount of flow pushed through the disjoint paths replacing edge $e'$ together and adding it to $e'$ in $G_{\Delta, n}$. The proof is straightforward and therefore omitted.
\end{proof}

Before we can carry out the proof of our lower bound, it remains to show for our lower bound gadget which is the graph  $G = G_{\Delta, n} \cup G^{\Phi}_{\Delta, n}$ that it is a $\Theta(\Phi)$-expander. 

\begin{claim}\label{clm:expanderAndFewEdges}
$G$ is $\Theta(\Phi)$-expander with $\Theta(n/\Phi^2)$ edges.
\end{claim}
\begin{proof}
The number of edges is straightforward from our construction of $G$. To see that $G$ is $O(\Phi)$-expander, observe that we can take the internal vertices of any path in $G^{\Phi}_{\Delta, n}$ replacing an edge in $G_{\Delta, n}$ which has volume $\Omega(1/\Phi)$ but only two edges leaving (the once to the endpoints of the replaced edges). To observe that it is an $\Omega(\Phi)$-expander, it suffices to show that each cut in $S$ is maximized by assigning all internal vertices of each such path to one side of the cut. It is then not hard to show from $G_{\Delta, n}$ being a $\Omega(1)$-expander that the claim follows.
\end{proof}

Let us now give the proof of the main result. We take the graph under consideration to be $G = G_{\Delta, n} \cup G^{\Phi}_{\Delta, n}$. We take as demand, the vector $\matr{\chi}^{\Phi} = \frac{1}{\Phi} \cdot \matr{\chi} = \frac{1}{\Phi} \cdot \{ \matr{\chi}_{(u,v)} \}_{e = (u,v) \in E(G_{\Delta, n})}$. Let $\matr{A}_{\mathcal{E}}$ denote the electrical flow routing on this graph $G$. Let us look at each edge $e = (u,v) \in  E(G_{\Delta, n})$. From Fact \ref{fact:constEffResistance}, Claim \ref{clm:gadgetAlsoHasGoodRes} and Fact \ref{fact:unionOfExpandersForRouting}, we have that the flow $\matr{f}_e = \matr{A}_{\mathcal{E}} \cdot \frac{1}{\Phi} \matr{\chi}_{(u,v)}$ restricted to the edges in $E(G_{\Delta, n})$ routes in total at least $\log n/ \phi$ units of flow along all of these edges. By linearity of $\matr{A}_{\mathcal{E}}$ and the fact that flows do not cancel, we have that when routing $\matr{\chi}^{\Phi}$, an average edge in $E(G_{\Delta, n})$ carries $\Omega(\log n/ \Phi)$ units of flow. Thus, the $\ell_p$-norm of this flow is at least $\sqrt[p]{n \cdot (\log n/ \Phi)^p} = n^{1/p} \cdot \log n \cdot \Phi$. But observe that we can route the flow with demand with congestion $1$ in  $G^{\Phi}_{\Delta, n}$ by routing for each demand $\matr{\chi}_{(u,v)}$ exactly $1$ unit of flow through each of the disjoint paths corresponding to the edge $e' = (u,v)$ in $G^{\Phi}_{\Delta, n}$. The $\ell_p$-norm of this flow is $\Theta((n /\Phi^2)^{1/p}) = \Theta(n^{1/p} \Phi^{-2/p})$ (using Claim \ref{clm:expanderAndFewEdges}). We thus have that $\rho_p(\matr{A}_{\mathcal{E}}) = \Omega(\log n \cdot \Phi^{p/(p-2)})$.

To obtain the result for $p < 2$, we use that for $q$ given by $1/p + 1/q = 1$, we have $\norm{\matr{X}}_{p \to p} = \norm{\matr{X^{\trp}}}_{q \to q}$, and for the electrical routing, $\matr{A}_{\mathcal{E}} \matr{B} = (\matr{A}_{\mathcal{E}} \matr{B})^{\trp}$ since $\matr{L}^+$ is symmetric. 

We note that in the construction above the number of vertices in the final graph $G$ might be much larger than $n$. By considering all possible parameters for $\Phi$ in $[1/n, 1]$ (i.e. all such numbers were $1/\Phi$ is integer), we obtain a family of $n'$-vertex graphs with conductances in $[1, 1/\sqrt[3]{n'}]$, as claimed. Since every $\Phi$-expander is also a $\Phi'$-expander for every $\Phi' \leq \Phi$, we do further not need to restrict the domain of $\Phi'$ further than in range. We point out that by considering parameters $\Phi$ in our construction that are even smaller than $1/n$, one can get up to an arbitrarily small polynomial factor close to conductances as small as $1/ \sqrt{n'}$. 

%% file: generalizednorms.tex
\section{Appendix -- Congestion for Monotonic Norms}
\label{sec:congestionAbsNorm}

This sections studies the competitive bound under more general norms, namely monotonic norms. 
We prove an identity for the competitive ratio in this case, which ultimately yields a way to efficiently compute it.
We include these bounds for completeness mainly.
The core ideas already appear in \cite{ER09,KLOS}.

We start by defining a much broader group of norms than previously considered when analyzing competitive bounds.

\begin{definition}
A norm \( \lVert \cdot \rVert \colon \R^n \to \R\) is called monotonic if for every \(\matr{x}, \matr{y} \in \R^n\) it holds that \(\lvert \matr{x} \rvert \leq \lvert \matr{y} \rvert \implies \lVert \matr{x} \rVert \leq \lVert \matr{y} \rVert\).
\end{definition}

\begin{definition}
    A norm \( \lVert \cdot \rVert \colon \R^n \to \R\) is called absolute if for every \(\matr{x}\in \R^n\) it holds that \(\lVert \matr{x} \rVert = \lVert \lvert \matr{x} \rvert \rVert\).
\end{definition}

\begin{fact}
\label{fact_monotonic_norms_are_absolute}
Every monotonic norm is also absolute.
\end{fact}

The main result of this section is the Lemma below. We point out that it works for a much stronger notion of congestion that considers weights. In the rest of the paper, we use the unweighted version of this theorem but prove it to full generality as a reference. 

\begin{lemma}[Competitive ratio of monotonic norms]
\label{lemma_competitive_ratio_as_norm_monotonic}
Let \(G = (V, E, w)\) be a graph with positive weights. For any oblivious routing \(\matr{A}\) and a definition of congestion with a monotonic norm \(\lVert \cdot \rVert\), we have:
\[
\rho(\matr{A})
=
\lVert \lvert \matr{W}^{-1} \matr{A} \matr{B} \matr{W} \rvert \rVert.
\]
Here, we denote by $\| \matr{X} \|$ where $\matr{X}$ is a matrix, the matrix norm induced by the vector norm $\| \cdot \|$, i.e. $\| \matr{X} \| = \sup_{\matr{z} \neq 0} \frac{\| \matr{X} \matr{z}\|}{\|\matr{z}\|}$.
\end{lemma}

\begin{proof}
Let \(D = \{\matr{\chi}_1, \dots, \matr{\chi}_k\}\) be a multiset of demands, and let \(\{\matr{f}_1', \dots, \matr{f}_k'\}\) be a choice of flows that route the demands in \(D\) optimally.

We will now construct a new multiset of demands that ``forces'' the routing described by \(\{\matr{f}_1', \dots, \matr{f}_k'\}\) on each edge. Concretely, given a multiset of demands \(D\), define the multiset \(D' = \{\matr{\chi}_i^e \, | \, i \in [k], e \in E\}\), where each \(\matr{\chi}_i^e\) corresponds to the amount of flow sent by \(\matr{f}_i\) on edge \(e\): 
\[ \matr{\chi}_i^{(a, b)}(u) = 
\begin{cases}
-\matr{f}_i'(a, b), & \text{if } u = a \\
\matr{f}_i'(a, b), & \text{if } u = b \\
0, & \text{otherwise}
\end{cases} \implies \matr{\chi}_i^{(a, b)} = \matr{f}_i'{(a, b)} \cdot \matr{\chi}_{(a, b)}.\]
Note that for any \(i \in [k]\), the demands \(\{\matr{\chi}_i^e\}_{e \in E}\) reconstruct \(\matr{\chi}_i\), that is, \(\sum_{e \in E}{\matr{\chi}_i^e} = \matr{\chi}_i\). 

In the following, we will prove \(\conge\left(\{\matr{A} \matr{\chi}\}_{\matr{\chi} \in D}\right) \leq \conge(\left\{\matr{A} \matr{\chi}\}_{\matr{\chi} \in D'}\right)\) and \(\opt(D) \geq \opt(D')\). These two inequalities will allow us to restrict our search domain of demand multisets for the computation of \(\rho(\matr{A})\). The motivation behind this restriction will become apparent after their proof, after which it will be straightforward to establish the lemma.

Proving \( \conge\left( \{\matr{A} \matr{\chi}\}_{\matr{\chi} \in D}\right) \leq \conge \left( \{\matr{A} \matr{\chi}\}_{\matr{\chi} \in D'}\right) \) can be done by employing the absolute monotonicity of the norm as follows:
\begin{equation}
    \label{equation_cong_d_leq_cong_dp}
    \begin{aligned}
    \conge\left(\{\matr{A} \matr{\chi}\}_{\matr{\chi} \in D}\right) &= \left\lVert \matr{W}^{-1} \cdot \sum_{i = 1}^k{\left\lvert\matr{A} \matr{\chi}_i \right\rvert} \right\rVert \\
    &= \left\lVert \matr{W}^{-1} \cdot \sum_{i = 1}^k{\left\lvert \sum_{e \in E}{\matr{A} \matr{\chi}_i^e} \right\rvert} \right\rVert \\
    &\leq \left\lVert \matr{W}^{-1} \cdot \sum_{i = 1}^k{\sum_{e \in E}{\left\lvert \matr{A} \matr{\chi}_i^e \right\rvert}} \right\rVert \\
    &= \conge\left(\{\matr{A} \matr{\chi}\}_{\matr{\chi} \in D'}\right)
    \end{aligned}
\end{equation}
The second part, namely \(\opt(D) \geq \opt(D')\), can be proven by showing that based on \(\{\matr{f}_1', \dots, \matr{f}_k'\}\) we can build a multiset of flows that satisfy \(D'\), and which has the same congestion. To see this, define for every \(e \in E \) the matrix \(\matr{I}_e = \textrm{diag}(\matr{1}_e)\) as the matrix that preserves only the component corresponding to the index of edge \(e\) of a vector upon multiplication from the left. Thus, we can write:
\rasmus{I might take a pass at phrasing here if I have time}
\begin{equation}
\begin{aligned}
\label{equation_rewrite_opt}
    \opt(D) &= \left\lVert \matr{W}^{-1} \cdot \sum_{i = 1}^k{\lvert \matr{f}_i'\rvert} \right\rVert \\
    &= \left\lVert \matr{W}^{-1} \cdot \sum_{i = 1}^k{\sum_{e \in E}{\left\lvert\matr{I}_e \matr{f}_i' \right\rvert}} \right\rVert \\
    &= \conge \left( \left\{\matr{I}_e \matr{f}_i'\right\}_{i \in [k], \, e \in E} \right).
\end{aligned}
\end{equation}
Observe now that for any \(i \in [k]\) and \(e \in E\) it holds by definition that \(\matr{B} \matr{I}_e \matr{f}_i' = \matr{\chi}_i^e\), which means that the flows in \(\{\matr{I}_e \matr{f}_i'\}_{i \in [k], \, e \in E}.\) fulfill the demands of \(D'\). Thus, we can conclude that \(\opt(D) \geq \opt(D')\).

We can use the result above and the one from Equation \eqref{equation_cong_d_leq_cong_dp} to restrict our domain of maximization for the competitive ratio. If we define \(\mathcal{D}\) to be the set of all multisets of demands on \(V\), and \(\mathcal{D'}\) to be the set of all multisets of demands that are obtained by splitting the optimal routing into isolated ``single-edge'' demands (as we did above for \(D\)), then we can write:
\begin{equation}
\label{equation_maxD_equals_maxDp}
\rho(\matr{A}) = \max_{D \in \mathcal{D}}{\frac{\conge\left(\{\matr{A} \matr{\chi}\}_{\matr{\chi} \in D}\right)}{\opt(D)}} = \max_{D' \in \mathcal{D'}}{\frac{\conge\left(\{\matr{A} \matr{\chi}\}_{\matr{\chi} \in D'}\right)}{\opt(D')}}.
\end{equation}
A useful observation is that the multiset of flows \(\{\matr{I}_e \matr{f}_i'\}_{i \in [k], \, e \in E}\) not only routes the demands in \(D'\), but it does so in an optimal manner, that is, \(\conge \left( \{\matr{I}_e \matr{f}_i'\}_{i \in [k], \, e \in E} \right) = \opt(D')\) . This means that, after constructing \(D'\) from a given \(D\), it will always be optimal to fulfill each of the demands in \(D'\) by routing flow only on a single edge.

To see why this holds, assume towards a contradiction that there exists a different collection of demand-fulfilling flows \(\{\matr{f}_{i, e}\}_{i \in [k], \, e \in E}\) for \(D'\) such that \(\conge \left( \{\matr{f}_{i, e}\}_{i \in [k], \, e \in E} \right) < \conge \left( \{\matr{I}_e \matr{f}_i'\}_{i \in [k], \, e \in E} \right)\). But since the norm is monotonic, we can obtain the following inequality from
\eqref{equation_rewrite_opt}:
\begin{align*}
    \opt(D) &= \left\lVert \matr{W}^{-1} \cdot \sum_{i = 1}^k{\sum_{e \in E}{\left\lvert\matr{I}_e \matr{f}_i' \right\rvert}} \right\rVert \\
    &= \conge \left( \{\matr{I}_e \matr{f}_i'\}_{i \in [k], \, e \in E} \right) \\
    &> \conge \left( \{\matr{f}_{i, e}\}_{i \in [k], \, e \in E} \right) \\
    &= \left\lVert \matr{W}^{-1} \cdot \sum_{i = 1}^k{\sum_{e \in E}{\left\lvert \matr{f}_{i, e} \right\rvert}} \right\rVert \\
    &\geq \left\lVert \matr{W}^{-1} \cdot \sum_{i = 1}^k{\left\lvert 
 \sum_{e \in E}{\matr{f}_{i, e}}\right\rvert} \right\rVert \\
    &= \conge \left( \left\{ \sum_{e \in E} \matr{f}_{i, e} \right\}_{i \in [k]} \right).
\end{align*}
Notice that, since \(\left\{ \matr{f}_{i, e} \right\}_{i \in [k], e \in E}\) routes the demands in \(D'\), by linearity it follows that \(\left\{ \sum_{e \in E} \matr{f}_{i, e} \right\}_{i \in [k]}\) will route the demands in \(D\). Thus, we obtained a contradiction to the initial assumption that \(\left\{ \matr{f}_1', \dots, \matr{f}_k' \right\}\) routes the demands of \(D\) optimally.

Now we can leverage the results above to prove:
\begin{equation}
    \label{equation_max_over_concrete_optima}
    \rho(\matr{A}) = \max_{D' \in \mathcal{D}'}{\frac{\conge\left(\{\matr{A} \matr{\chi}\}_{\matr{\chi} \in D'}\right)}{\opt(D')}} = \max_{\matr{x} \in \R^E}{\frac{\left\lVert \matr{W}^{-1} \cdot \sum_{e \in E}{\left( \matr{x}(e) \cdot \lvert \matr{A} \matr{\chi}_e \rvert \right)} \right\rVert}{\left\lVert \matr{W}^{-1} \matr{x} \right\rVert}}.
\end{equation}
This will proceed in two steps. We first show that the LHS is not bigger than the RHS and then vice-versa, which ultimately implies equality.

To prove that the LHS is at most the RHS, recall that our previously defined \(D\) was arbitrary in the set \(\mathcal{D}'\), therefore it suffices to show that there exists \(\matr{x} \in \R^E\) such that the two ratios in Equation \eqref{equation_max_over_concrete_optima} are equal. Since, we have already shown that \(\opt(D') = \conge \left( \{\matr{I}_e \matr{f}_i'\}_{i \in [k], \, e \in E} \right)\), choose \(\matr{x} = \sum_{i = 1}^k{\sum_{e \in E}{\left\lvert \matr{I}_e \matr{f}_i' \right\rvert}}\), which trivially satisfies the equality between the denominators. Note that this pick for \(\matr{x}\) gives:
\[\matr{x} = \sum_{i = 1}^k{\sum_{e \in E}{\left\lvert \matr{I}_e \matr{f}_i' \right\rvert}} \implies \matr{x}(e) = \sum_{i = 1}^k{\lvert \matr{f}_i'(e) \rvert} \text{, for any } e \in E.\]
In other words, \(\matr{x} \in \R^E\) is the vector that collects the absolute values of the flows routed on each edge by the multiset of flows \(\left\{ \matr{I}_e \matr{f}_i' \right\}_{i \in [k], \, e \in E}\). By definition, these flows route the demands in \(D'\), which means that we can rewrite the congestion caused by \(\matr{A}\) as follows:
\begin{equation}
\begin{aligned}
\label{equation_congestion_of_edge_demands}
    \conge\left(\{\matr{A} \matr{\chi}\}_{\matr{\chi} \in D'}\right) &= \left\lVert \matr{W}^{-1} \cdot \sum_{\matr{\chi} \in D'}{\left\lvert \matr{A} \matr{\chi} \right\rvert} \right\rVert \\
    &= \left\lVert \matr{W}^{-1} \cdot \sum_{i = 1}^k{\sum_{e \in E}{\left\lvert \matr{A} \matr{\chi}_i^e \right\rvert}} \right\rVert \\
    &= \left\lVert \matr{W}^{-1} \cdot \sum_{i = 1}^k{\sum_{e \in E}{\left\lvert \matr{A} \cdot \left( \matr{f}_i'(e) \cdot \matr{\chi}_e \right) \right\rvert}} \right\rVert \\
    &= \left\lVert \matr{W}^{-1} \cdot \sum_{i = 1}^k{\sum_{e \in E}{\left(\left\lvert \matr{f}_i'(e) \right\rvert \cdot \left\lvert \matr{A}  \matr{\chi}_e \right\rvert \right) }} \right\rVert \\
    &= \left\lVert \matr{W}^{-1} \cdot \sum_{e \in E}{\left( \left\lvert \matr{A}  \matr{\chi}_e \right\rvert \cdot \sum_{i = 1}^k{\left(\left\lvert \matr{f}_i'(e) \right\rvert \right)} \right)} \right\rVert \\
    &= \left\lVert \matr{W}^{-1} \cdot \sum_{e \in E}{\left( \left\lvert \matr{A}  \matr{\chi}_e \right\rvert \cdot \matr{x}(e) \right)} \right\rVert.
\end{aligned}
\end{equation}
This finally gives the equality between the numerators, which concludes the first part of the proof of the equality in \eqref{equation_max_over_concrete_optima}. In order to complete the proof, we now show that for any \(\matr{x} \in \R^E\) one can construct \(D_0 \in \mathcal{D}\) such that:
\begin{equation}
\label{equation_intermediate_inequality}
\frac{\left\lVert \matr{W}^{-1} \cdot \sum_{e \in E}{\left( \matr{x}(e) \cdot \lvert \matr{A} \matr{\chi}_e \rvert \right)} \right\rVert}{\left\lVert \matr{W}^{-1} \matr{x} \right\rVert} \leq \frac{\conge\left(\{\matr{A} \matr{\chi}\}_{\matr{\chi} \in D_0}\right)}{\opt(D_0)}.
\end{equation}
Note that, due to \eqref{equation_maxD_equals_maxDp}, proving the statement above is sufficient to conclude the proof. To this extent, consider an arbitrary vector \(\matr{x} \in \R^E\), and let \(D_0 = \left\{ \lvert \matr{x}(e) \rvert \cdot \matr{\chi}_e \right\}_{e \in E}\). Due to the weights being positive and the norm being monotonic, we can obtain the following inequality between the target numerators:
\begin{align*}
    \left\lVert \matr{W}^{-1} \cdot \sum_{e \in E}{\left( \matr{x}(e) \cdot \lvert \matr{A} \matr{\chi}_e \rvert \right)} \right\rVert &\leq \left\lVert \matr{W}^{-1} \cdot \sum_{e \in E}{\left(\lvert \matr{x}(e) \rvert \cdot \lvert \matr{A} \matr{\chi}_e \rvert \right)} \right\rVert \\
    &= \left\lVert \matr{W}^{-1} \cdot \sum_{e \in E}{\left\lvert \matr{A} \left(\lvert \matr{x}(e) \rvert \cdot \matr{\chi}_e \right) \right\rvert } \right\rVert \\
    &= \conge\left(\{\matr{A} \cdot \matr{\chi}\}_{\matr{\chi} \in D_0}\right).
\end{align*}
Observe now that for any \(e \in E\) it holds that \(\matr{B} \matr{I}_e \lvert \matr{x} \rvert = \lvert \matr{x}(e) \rvert \cdot \matr{\chi}_e\), therefore the set of flows \(\left\{ \matr{I}_e \lvert \matr{x} \rvert \right\}_{e \in E}\) fulfills the demands in \(D_0\). We can thus use Fact \ref{fact_monotonic_norms_are_absolute} to obtain the following inequality for the denominators:
\[\left\lVert \matr{W}^{-1} \matr{x} \right\rVert = \left\lVert \matr{W}^{-1} \lvert \matr{x} \rvert \right\rVert = \left\lVert \matr{W}^{-1} \cdot \sum_{e \in E}{\lvert \matr{I}_e \lvert \matr{x} \rvert \rvert } \right\rVert = \conge \left( \left\{ \matr{I}_e \lvert \matr{x} \rvert \right\}_{e \in E}\right) \geq \opt(D_0).\]
The previous two inequalities conclude the proof of the statement in \eqref{equation_intermediate_inequality}.

The last paragraphs focused on showing that the RHS of \eqref{equation_max_over_concrete_optima} is at most the LHS, which ultimately proves the desired equality (since the other direction of the inequality had been proven previously).

We now return to Equation \eqref{equation_max_over_concrete_optima} with the aim of rewriting it to obtain the sought form. To that extent, note that the incidence matrix \(\matr{B}\) is the matrix whose columns are all \(\matr{\chi}_e\) with \(e \in E\). From this fact and the assumption that all weights are positive, we can rewrite \eqref{equation_max_over_concrete_optima} to get the desired equality:
\begin{align*}
    \rho(\matr{A}) &= \max_{\matr{x} \in \R^E}{\frac{\left\lVert \matr{W}^{-1} \cdot \sum_{e \in E}{\left( \matr{x}(e) \cdot \lvert \matr{A} \matr{\chi}_e \rvert\right)} \right\rVert}{\left\lVert \matr{W}^{-1} \matr{x} \right\rVert}} \\
    &= \max_{\matr{x} \in \R^E}{\frac{\left\lVert\sum_{e \in E}{\left( \matr{x}(e) \cdot \lvert \matr{W}^{-1} \matr{A} \matr{\chi}_e \rvert \right)} \right\rVert}{\left\lVert \matr{W}^{-1} \matr{x} \right\rVert}} \\
    &= \max_{\matr{x} \in \R^E}{\frac{\left\lVert \lvert \matr{W}^{-1} \matr{A} \matr{B} \rvert \matr{x} \right\rVert}{\left\lVert \matr{W}^{-1} \matr{x} \right\rVert}} \\
    &= \max_{\matr{x} \in \R^E}{\frac{\left\lVert \lvert \matr{W}^{-1} \matr{A} \matr{B} \rvert \matr{W} \matr{W}^{-1} \matr{x} \right\rVert}{\left\lVert \matr{W}^{-1} \matr{x} \right\rVert}} \\
    &= \max_{\matr{y} \in \R^E}{\frac{\left\lVert \lvert \matr{W}^{-1} \matr{A} \matr{B} \rvert \matr{W} \matr{y} \right\rVert}{\left\lVert \matr{y} \right\rVert}} \\
    &= \max_{\matr{y} \in \R^E}{\frac{\left\lVert \lvert \matr{W}^{-1} \matr{A} \matr{B} \matr{W} \rvert \matr{y} \right\rVert}{\left\lVert \matr{y} \right\rVert}} \\
    &= \left\lVert \lvert \matr{W}^{-1} \matr{A} \matr{B} \matr{W} \rvert\right\rVert.
\end{align*}
\end{proof}
